\documentclass{article}
\usepackage{amssymb}
\usepackage{amsthm}
\usepackage{hyperref}
\usepackage{tikz}
\usetikzlibrary{automata}
\newtheorem{thm}{Theorem}
\newtheorem{defn}{Definition}
\newtheorem{prop}[thm]{Proposition}
\newtheorem{lemma}[thm]{Lemma}
\newtheorem{example}[thm]{Example}

\title{EF+EX Forest Algebras}
\author{Andreas Krebs, University of T\"ubingen\\
Howard Straubing, Boston College}

\begin{document}
\maketitle

\begin{abstract}
We examine languages of unranked forests definable using the temporal operators EF and EX. We characterize the languages definable in this logic, and various fragments thereof, using the syntactic forest algebras introduced by Bojanczyk and Walukiewicz. Our algebraic characterizations yield efficient algorithms for deciding when a given language of forests is definable in this logic.  The proofs are based on understanding the wreath product closures of  a few small algebras, for which we introduce a general ideal theory for forest algebras.
This combines ideas from the work of Bojanczyk and Walukiewicz for the analogous logics on binary trees and from early work of Stiffler on wreath product of finite semigroups.
\end{abstract}

\section{Overview}
Understanding the expressive power of temporal and first-order logic on trees is important in several areas of computer science, for example in formal verification. Using algebraic methods, in particular, finite monoids, to understand the power of subclasses of the regular languages of finite words has proven to be extremely successful, especially in the characterization of regular languages definable in various fragments of first-order and temporal logics~(\cite{CPP,TW,STR}). Here we are interested in sets of of finite trees (or, more precisely, sets of finite forests), where the analogous algebraic structures are forest algebras.

Bojanczyk {\it et. al.}~\cite{BW2,BSW} introduced forest algebras, and underscored the importance of the wreath product decomposition theory of these algebras in the study of the expressive power of temporal and first-order logics on finite unranked trees. For languages inside of $\mathsf{CTL}$ the associated forest algebras can be built completely via the wreath product of copies of the forest algebra $${\cal U}_2 = (\{0,\infty\},\{1,0,c_0\}),$$ where the vertical element 0 is the constant map to $\infty,$ and the vertical element $c_0$ is the constant map to 0 (\cite{BSW}).  The problem of effectively characterizing the wreath product closure of ${\cal U}_2$ is thus an important open problem, equivalent to characterization of  $\mathsf{CTL}.$ Note that if one strips away the additive structure of ${\cal U}_2,$ the wreath product closure is the family of all finite aperiodic semigroups (the Krohn-Rhodes Theorem). Forest algebras have been successfully applied to the obtain characterization of other logics on trees; see,  for example~\cite{BSS,BS}.

Here we study in detail the wreath product closures of proper subalgebras of ${\cal U}_2$.  In one sense, this generalizes early work of Stiffler~\cite{STI}, who carried out an analogous program for wreath products of semigroups. Along the way, we develop the outlines of a general ideal theory for forest algebras, which we believe will be useful in subsequent work.
After developing the algebraic theory, we give an application to logic.  We obtain a characterization of the languages of unranked forests definable using  the temporal operators $\mathsf{EF}$ and $\mathsf{EX}.$  This extends the work of Bojanczyk and Walukiewicz~\cite{BW1}, who obtain a similar characterization for the analogous logics on binary trees.  Our proof, however, which proceeds entirely from the algebraic analysis, is completely different.
Similar results, again for the case of trees of bounded rank, appear in \'Esik \cite{ESIK}.

The paper is structured in the following way. First we introduce forest algebras and introduce some general theory (see Section \ref{sec.forestalgebras}). After giving the connections between sublogics of $\mathsf{CTL}$ in Section~\ref{section.logic},  we examine in detail algebraic operations corresponding to closure under the $\mathsf{EF}$ quantifier (Section~\ref{sec.ef}), the $\mathsf{EX }$ quantifier (Section~\ref{sec.ex}) and then both quantifiers (Section~\ref{sec.efex}). We conclude with our characterization and decidability results %
in Section~\ref{sec.con}).
\section{Forest Algebras}\label{sec.forestalgebras}
\subsection{Preliminaries}
We refer the reader to ~\cite{BW2,BSW} for the definitions of abstract forest algebra, free forest algebra, and syntactic forest algebra. We denote the free forest algebra over a finite alphabet $A$ by $A^{\Delta}=(H_A,V_A),$ where $H_A$ denotes the monoid of forests over $A,$ with concatenation as the operation, and $V_A$ denotes the monoid of contexts over $A,$ with composition as the operation. A subset $L$ of $H_A$ is called a forest language over $A.$ We denote its syntactic forest algebra by $(H_L,V_L),$ and its syntactic morphism by 
$$\mu_L:A^{\Delta}\to (H_L,V_L).$$

For the most part, our principal objects of study are not the forest algebras themselves, but homomorphisms
$$\alpha:A^{\Delta}\to (H,V).$$
It is important to bear in mind that each such homomorphism is actually a pair of monoid homomorphisms, one mapping $H_A$ to $H$ and the other mapping $V_A$ to $V.$ It should usually be clear from the context which of the two component homomorphisms we mean, and thus we denote them both by $\alpha.$  The `freeness' of $A^{\Delta}$ is the fact that a homomorphism $\alpha$ into $(H,V)$ is completely determined by giving its value, in $V,$ at each $a\in A.$ 

A homomorphism $\alpha$ as above {\it recognizes} a language $L\subseteq H_A$ if there exists $X\subseteq H$ such that $\alpha^{-1}(X)=L.$

If $\alpha:A^{\Delta}\to (H,V)$ and  $\beta:A^{\Delta}\to (H',V'),$ are homomorphisms, we say that $\beta$ {\it factors through} $\alpha$ if for all $s,s'\in H_A,$ $\alpha(s)=\alpha(s')$ implies $\beta(s)=\beta(s').$  This is equivalent to the existence of a homomorphism $\rho$ from the image of $\alpha$ into $(H',V')$ such that $\beta=\rho\alpha.$  A homomorphism $\alpha$ recognizes $L\subseteq H_A$ if and only if $\mu_L$ factors through $\alpha.$(~\cite{BW2}).

In the course of the paper we will see several {\it congruences} defined on free forest algebras.  Such a congruence is determined by an equivalence relation $\sim$ on $H_A$ such that for any $p\in V_A,$ $s\sim s'$ implies $ps\sim ps'.$ This gives a well-defined action of $V_A$ on the set of $\sim$-classes of $H_A.$  We define an equivalence relation (also denoted $\sim$) on $V_A$ by setting $p\sim p'$ if for all $s\in H_A,$ $ps\sim p's.$  The result is a quotient forest algebra $(H_A/{\sim},V_A/{\sim}).$  In order to prove that an equivalence relation $\sim$ on $H_A$ is a congruence, it is sufficient to verify that $s\sim s'$ implies $s+t\sim s+t'$ and $as \sim as'$ for all $s,s',t\in H_A$ and $a\in A.$

\subsection{Horizontally idempotent and commutative algebras}
We now introduce an important restriction.  Throughout the rest of the paper, we will assume that  all of our finite forest algebras $(H,V)$ have $H$ idempotent and commutative; that is $h+h=h'+h$ and $h+h=h$ for all $h,h'\in H.$  This is a natural restriction when talking about classes of forest algebras arising in temporal logics, which is the principal application motivating this study.

When $H$ is horizontally idempotent and commutative, the sum of all its elements is an absorbing element for the monoid.  While an absorbing element in a monoid is ordinarily written 0, since we use additive notation for $H,$ its identity is denoted 0, and accordingly we denote the absorbing element, which is necessarily unique, by $\infty.$

We say that two forests $s_1,s_2\in H_A$ are {\it idempotent-and-commutative equivalent} if $s$ can be transformed into $t$ by a sequence of operations of the following three types: {\it (i)} interchange the order of two adjacent subtrees (that is, if $s=p(t_1+t_2)$ for some context $p$ and trees $t_1,t_2,$ then we transform $s$ to $p(t_2+t_1)$); {\it (ii)} replace a subtree $t$ by two adjacent copies (that is, transform $pt$ to $p(t+t)$); {\it (iii)} replace two identical adjacent subtrees by a single copy (transform $p(t+t)$ to $pt$).  Since operations {\it (ii)} and {\it (iii)} are inverses of one another, and operation {\it (i)} is its own inverse, this is indeed an equivalence relation.

We have the following obvious lemma:
\medskip\begin{lemma}\label{lemma.ic_equivalence}
Let $\alpha:A^{\Delta}\to (H,V)$ be a homomorphism, where $H$ is horizontally idempotent and commutative.  If $s,t\in H_A$ are idempotent-and-commutative equivalent, then $\alpha(s)=\alpha(t).$
\end{lemma}
\begin{proof} By idempotence and commutativity of $H,$ each of the three operation types used to transform $s$ into $t$ preserves the value under $\alpha.$
\end{proof}

There is a smallest nontrivial idempotent and commutative forest algebra,  ${\cal U}_1=(\{0,\infty\},\{1,0\}).$ The horizontal and vertical monoids of  ${\cal U}_1$ are isomorphic, but we use different names for the elements because of the additive notation for the operation in one of these monoids, and multiplicative notation in the other.  We have not completely specified how 
the vertical monoid acts on the horizontal monoid---this is done by setting $0\cdot x=\infty$ for $x\in\{0,\infty\}.$

\subsection{1-definiteness}

In Section~\ref{sec.ex} we will discuss in detail the notion of definiteness in forest algebras; for this preliminary section, we will only need to consider a special case.  A forest algebra homomorphism $\alpha:A^{\Delta}\to (H,V)$ is said to be {\it 1-definite} if for $s\in H_A,$ the value of $\alpha(s)$ depends only on the set of labels of the root nodes of $s.$  We define an equivalence relation $\sim_1$ on $H_A$ by setting $s\sim_1 s'$ if and only if the sets of labels of root nodes of $s$ and $s'$ are equal.  This defines a congruence on $A^{\Delta}.$  We denote the homomorphism from $A^{\Delta}$ onto the quotient under $\sim_1$ by $\alpha_{A,1}.$  It is easy to show that a homomorphism $\alpha:A^{\Delta}\to (H,V)$ is 1-definite if and only if it factors through $\alpha_{A,1}.$

\subsection {Wreath Products}
We summarize the discussion of wreath products given in ~\cite{BSW}.
The {\it wreath product} of two forest algebras $(H_1,V_1), (H_2,V_2)$ is
$$(H_1,V_1)\circ (H_2,V_2)=(H_1\times H_2,V_1\times V_2^{H_1}),$$
where the monoid structure of $H_1\times H_2$ is the ordinary direct product, and the action is given by
$$(v_1,f)(h_1,h_2)=(v_1h_1,f(h_1)h_2),$$
for all $h_1\in H_1,$ $h_2\in H_2,$ $v_1\in V_1,$ and $f:H_1\to V_2.$  It is straightforward to verify that the resulting structure satisfies the axioms for a forest algebra. Note that if one forgets about the monoid structure on $H_1$ and $H_2,$ this is just the ordinary wreath product of left transformation monoids.  Because we use left actions rather than the right actions that are traditional in the study of monoid decompositions, we reverse the usual order of the factors.  The projection maps
$$\pi:(h_1,h_2)\mapsto h_1, (v,f)\mapsto v,$$
define a homomorphism from the wreath product onto the left-hand factor. 

We will view wreath products through the lens of homomorphisms from the free forest algebra.  Given such a homomorphism
$$\gamma:A^{\Delta}\to (H_1,V_1)\circ (H_2,V_2)$$
we can write, for each $a\in A,$
$$\gamma(a)=(v_a,f_a).$$
This gives rise to a pair of homomorphisms
$$\alpha:A^{\Delta}\to (H_1,V_1),\beta:(A\times H_1)^{\Delta}\to (H_2,V_2),$$
where $\alpha(a)=v_a,$ and $\beta(a,h)=f_a(h).$  We write $\gamma=\alpha\otimes\beta.$  Note that $\alpha=\pi\gamma,$ where $\pi$ is the projection onto the left-hand factor.  Conversely, any pair of homomorphisms $\alpha$ and $\beta$ as above gives rise to a homomorphism $\alpha\otimes\beta$ into the wreath product. We then have, for any $s\in H_A,$
$$\alpha\otimes\beta(s)=(\alpha(s), \beta(s^{\alpha})),$$
where $s^{\alpha}\in H_{A\times H_1}$ is obtained from $s$ through a relabeling process:  if a node $x$ of $s$ is originally labeled $a\in A,$ and the tree rooted at $x$ is $at,$ where $t\in H_A,$ then the label of the same node in $s^{\alpha}$ is $(a,\alpha(t)).$

The wreath product is an associative operation on forest algebras.  Given forest algebras $(H_i,V_i),$ $i=1,\ldots,r,$ and homomorphisms
$$\alpha_1:A^{\Delta}\to (H_1,V_1),$$
$$\alpha_i:(A\times H_1\times\cdots H_{i-1})^{\Delta}\to (H_i,V_i),$$
for $i=2,\ldots,r,$
we can form the homomorphism
$$\alpha_1\otimes\cdots\otimes\alpha_r:A^{\Delta}\to (H_1,V_1)\circ\cdots\circ (H_r,V_r).$$

A homomorphism
$$\alpha:A^{\Delta}\to (H_,V_1)\times (H_2,V_2)$$
into a direct product factors through the wreath product in a trivial way:  Let $\alpha_1,\alpha_2$ be the two component homomorphisms, and set $\beta(a,h)=\alpha_2(a)$ for all $a\in A, h\in H_1.$  Then $\alpha$ factors through $\alpha_1\otimes\beta.$

\subsection{Reachability}
Let $(H,V)$ be a finite forest algebra. For $h,h'\in H$ we write $h\leq h'$ if $h=vh'$ for some $v\in V,$  and say that $h$ is {\it reachable} form $h'.$ This gives a preorder on $H.$  We set $h\cong h' $ if both $h\leq h'$ and $h'\leq h.$ An equivalence class of $\cong$ is called a {\it reachability class}. The preorder consequently results in a partial order on the set of reachability classes of $H.$   We always have $h+h'\leq h,$ because $h+h'=(1+h')h.$ If $h\in H$ and $\Gamma$ is a reachability class of $H$ then we write, for example, $h\geq \Gamma$ to mean that $\Gamma\leq \Gamma',$ where $\Gamma'$ is the class of $h.$ 

A {\it reachability ideal} in $(H,V)$ is a subset $I$ of $H$ such that $h\in I$ and $h'\leq h$ implies $h'\in I.$  If we have a homomorphism
$$\alpha:A^{\Delta}\to (H,V)$$
and a reachability ideal $I\subseteq H,$ we define an equivalence relation ${\sim}_I$ on $H_A$ by setting $s\sim_I s'$ if $\alpha(s)=\alpha(s')\notin I,$ or if $\alpha(s),\alpha(s')\in I.$ Easily $s\sim_I s'$ implies $ps\sim_I ps'$ for any $p\in V_A.$  We thus obtain a homomorphism onto the quotient algebra
$$\alpha_I:A^{\Delta}\to (H/{\sim}_I,V/{\sim}_I)$$
which factors through $\alpha.$  Note that $I$ is, in particular, a two-sided ideal in the monoid $H,$ and $H/{\sim_I}$ is identical to the usual quotient monoid $H/I=(H-I)\cup\{\infty\}.$ We will thus use the notation $(H/I,V/I)$ for the quotient algebra, instead of $(H/{\sim}_I,V/{\sim}_I).$  If $\Gamma\subseteq H$ is a reachability class, then both
$$I_\Gamma=\{h\in H: h\not >\Gamma\}
\mbox{ and }
I_{{\mbox{\tiny$\geq$}}\Gamma}=\{h\in H: h\not\geq\Gamma\}$$
are reachability ideals.  We denote the associated quotients and projection homomorphisms by $(H_{\Gamma},V_{\Gamma}),$
$\alpha_{\Gamma},$ $(H_{{\mbox{\tiny$\geq$}}\Gamma},V_{{\mbox{\tiny$\geq$}}\Gamma})$, $\alpha_{{\mbox{\tiny$\geq$}}\Gamma}.$

Given the restriction that $H$ is idempotent and commutative, the absorbing element $\infty$ is reachable from every element.  The reachability class of $\infty$ is accordingly the unique minimal class, which we denote $\Gamma_{\mathsf{min}}.$ 
A reachability class $\Gamma$ is {\it subminimal} if $\Gamma_{\mathsf{min}}<\Gamma,$ but there is no class $\Lambda$ with $\Gamma_{\mathsf{min}}<\Lambda<\Gamma.$  The following lemma will be used several times.

\medskip\begin{lemma}\label{lemma.subminimal}
Let $\alpha:A^{\Delta}\to (H,V),$ and
let $\Gamma_1,\ldots,\Gamma_r$ be the subminimal reachability classes of $(H,V).$  Then 
$$\alpha_{\Gamma_{\mathsf{min}}}:A^{\Delta}\to (H_{\Gamma_{\mathsf{min}}},V_{\Gamma_{\mathsf{min}}})$$
factors through the direct product
$$\biggl(\prod_{j=1}^r\alpha_{{\mbox{\tiny$\geq$}}\Gamma_j}\biggr):A^{\Delta}\to \prod_{j=1}^r(H_{{\mbox{\tiny$\geq$}}\Gamma_j},V_{{\mbox{\tiny$\geq$}}\Gamma_j}).$$
Further each of the algebras $(H_{{\mbox{\tiny$\geq$}}\Gamma_j},V_{{\mbox{\tiny$\geq$}}\Gamma_j})$ has a unique subminimal reachability class.
\end{lemma}
\begin{proof}Choose $s,s'\in H_A$ such that $\alpha_{{\mbox{\tiny$\geq$}}\Gamma_j}(s)=\alpha_{{\mbox{\tiny$\geq$}}\Gamma_j}(s')$ for all $j=1,\ldots,r.$ If $\alpha(s)>\Gamma_{\mathsf{min}},$ then $\alpha(s)\geq\Gamma_j$ for some $j,$ and thus $\alpha(s)=\alpha(s'),$ so in particular $\alpha_{\Gamma_{\mathsf{min}}}(s)=\alpha_{\Gamma_{\mathsf{min}}}(s').$  If $\alpha(s)\not >\Gamma_{\mathsf{min}},$ then $\alpha(s)\in\Gamma_{\mathsf{min}},$  by minimality. Thus every $\alpha_{{\mbox{\tiny$\geq$}}\Gamma_j}(s)$ is the absorbing element $\infty$ of the quotient algebra, so the same is true for $\alpha_{{\mbox{\tiny$\geq$}}\Gamma_j}(s').$  Thus $\alpha(s')\not\geq \Gamma_j$ for all $j,$ so $\alpha(s')\in\Gamma_{\mathsf{min}},$ and $\alpha_{\Gamma_{\mathsf{min}}}(s)=\alpha_{\Gamma_{\mathsf{min}}}(s').$  This proves the claim about factorization of the homomorphisms.  

Next, for the claim about the subminimal classes of $(H_{{\mbox{\tiny$\geq$}}\Gamma_j},V_{{\mbox{\tiny$\geq$}}\Gamma_j}),$  observe that the reachability classes of this algebra are just the reachability classes of $(H,V)$ that are greater than or equal to $\Gamma_j,$ along with the minimal class $\{\infty\}.$ 
\end{proof}

We will also need the following lemma, which concerns the behavior of reachability classes under homomorphisms.
\medskip\begin{lemma}\label{lemma.morphisms}
Let $\beta:(H_1,V_1)\to (H_2,V_2)$ be a homomorphism of finite forest algebras.  Let $\Lambda\subseteq H_1$ be a reachability class.
There is a reachability class $\Gamma$ of $(H_2,V_2)$ such that $\beta(\Lambda)\subseteq\Gamma.$
If $\Lambda$ is a minimal class of $(H_1, V_1)$ satisfying  $\beta(\Lambda)\subseteq\Gamma,$ and $\beta$ is onto, then $\beta(\Lambda)=\Gamma.$
If, further, $H_2$ is idempotent and commutative, then there is only one such minimal class $\Lambda.$
\end{lemma}
\begin{proof}
Let $h_1,h_1'\in\Lambda,$ and let $h_2=\beta(h_1),$ $h_2'=\beta(h_1').$ To prove the first claim, we must show $h_2\cong h_2'.$ There exist $v,v'\in V_1$ such that $h_1=v'h_1',$ $h_1'=vh_1.$  We then have $h_2=\beta(v')h_2'$ and $h_2'=\beta(v')h_2,$ which gives the result.

Now suppose $\Gamma$ is the class of $(H_2,V_2)$ containing $\beta(\Lambda),$ and that $\Lambda$ is a minimal class in the preimage of $\Gamma.$  Let $h\in \Lambda,$ and let $h'\in \Gamma.$ We need to show $h'\in \beta(\Lambda).$  We have $v,v'\in V_2$ such that $\beta(h)=v'h',$ $h'=v\beta(h).$  Since $\beta$ is onto, there are elements $u,u'\in V_1$ with $\beta(u)=v,$ $\beta(u')=v'.$  We then have 
$\beta(h)=v'v\beta(h)=\beta(u'uh),$
but this means the class of $u'uh$ maps into $\Gamma.$  By minimality $u'uh\in\Lambda,$, and since $u'uh\leq uh\leq h,$ $uh\in\Lambda.$  Thus $h'=\beta(uh)\in\beta(\Lambda).$

For the last claim, suppose to the contrary that $h,h'\in V_1$ are both in minimal classes mapping into $\Lambda,$ but are not in the same class. By idempotence, $\beta(h+h') =\beta(h)\in\Gamma.$  We thus have $h+h'\leq h,$ so by minimality $h+h'\cong h.$ Likewise, $h+h'\cong h',$ so $h\cong h',$ a contradiction.
\end{proof}

\section{Connections to Logic}\label{section.logic}
\subsection{Temporal logics for forests}

We give a description of the temporal operators $\mathsf{EF}$ and $\mathsf{EX}.$  Our approach closely follows the one given in ~\cite{BSW}.

We describe the syntax and semantics of our formulas, given a fixed finite alphabet $A.$  One complication, which does not seem to be avoidable, is that we need to treat trees and forests somewhat differently.  Thus we define both tree formulas, and a proper subset of these called forest formulas, and give different semantics depending whether we are interpreting a formula in a tree or in a forest.

\begin{itemize}
\item {\bf T} is a forest formula
\item $a$ is a tree formula, for each $a\in A.$
\item every forest formula is a tree formula
\item both the class of tree formulas and the class of forest formulas are closed under boolean operations
\item if $\phi$ is a tree formula, then $\mathsf{EF}\phi,$ $\mathsf{EX}\phi$ are forest formulas
\end{itemize}

The semantics are similarly defined by mutual recursion.  There are two satisfaction relations, one for trees satisfying tree formulas, the other for forests satisfying forest formulas.

\begin{itemize}
\item $s\models_f {\bf T}$ for every forest $s\in H_A.$
\item $as\models_t a$ for every forest $s\in H_A.$
\item if $\phi$ is a forest formula, then $as\models_t \phi$ if and only if $s\models_f\phi.$
\item boolean operations have their usual interpretation
\item $s\models_f\mathsf{EF}\phi$ iff $s_x\models_t\phi,$ where $s_x$ denotes the tree rooted at some node $x$ of $s.$
\item $s\models_f\mathsf{EX}\phi$ iff $s_x\models_t\phi,$ where $x$ is a {\it root node} of $s.$
\end{itemize}

We'll call this logic $\mathsf{EF}+\mathsf{EX},$ and denote by $\mathsf{EF},$ $\mathsf{EX}$ the fragments in which only one of the two operators is used.

Intuitively, when we interpret formulas in trees, $\mathsf{EF}\phi$ means `at some time in the future $\phi$' and $\mathsf{EX}\phi$ means `at some next time $\phi$'. When we interpret such formulas in forests, we are in a sense treating the forest as though it were a tree with a phantom root node.  Observe that if $a\in A,$ we do not interpret the formula $a$ in forests at all.  Thus a formula can have different interpretations depending on whether we view it as a tree or a forest formula.  For example, as a forest formula $\mathsf{EX}a$ means `there is a root node labeled $a$' while as a tree formula it means `some child of the root is labeled $a$'.

We are primarily concerned with the forest satisfaction relation, and so we will usually drop the subscript on $\models,$ and assume that $\models_f$ is intended. If $\phi$ is a forest formula, then we denote by $L_{\phi}$ the set of all $s\in H_A$ such that $s\models\phi.$  $L_{\phi}$ is the {\it language defined by $\phi.$}

\begin{example}\label{example.efexformula}  Consider the following property of forests over $\{a,b\}$:  There is a tree component containing only $a$'s, and another tree component that contains at least one $b.$  Now consider the set $L$ of forests  $s$ that either have this property, or in which for some node $x,$ the forest of strict descendants of $x$ has the property.  The property itself is defined by the forest formula
$$\psi: \mathsf{EX}(a\wedge\neg \mathsf{EF}b)\wedge\mathsf{EX}(b\vee\mathsf{EF}b)$$
and $L$ is defined by 
$$\psi\vee\mathsf{EF}\psi.$$
In Example~\ref{example.mainexample}, we discuss the syntactic forest algebra of $L.$
\end{example}

\subsection{Correspondence of operators with wreath products}

The principal result of this paper, Theorem~\ref{thm.summary}, is the algebraic characterization of the forest languages using the operators $\mathsf{EF}$ and $\mathsf{EX},$ either separately or in combination. It will require some algebraic  preparation, in Sections~\ref{sec.ef}, \ref{sec.ex}  and~\ref{sec.efex} before we can give the precise statement of this theorem.  The bridge between the logic and the algebra is provided by the next two propositions.

Let $\phi$ be a tree formula.  Then $\phi$ can be written as a disjunction
$$\bigvee_{a\in A}(a\wedge\psi_a),$$
where each $\psi_a$ is a forest formula.  Let $\Psi=\{\psi_a:a\in A\}.$ We'll call $\Psi$ the set of forest formulas of $\phi.$  We say that a homomorphism 

$$\beta:A^{\Delta}\to (H,V)$$
recognizes $\Psi$ if the value of $\beta(s)$ determines exactly which formulas of $\Psi$ are satisfied by $s.$  To construct such a homomorphism, we can take the direct product of the syntactic algebras of $L_{\psi}$ for $\psi\in\Psi,$ and set $\beta$ to be the product of the syntactic morphisms.

The following theorem, adapted from ~\cite{BSW}, gives the connection between the $\mathsf{EF}$ operator and wreath products with ${\cal U}_1$:

\medskip\begin{prop}\label{prop.efoperator}
\noindent{(a)}  Suppose that $\phi$ is a tree formula, $\Psi$ is the set of forest formulas of $\phi,$ and that $\Psi$ is recognized by 
$$\alpha:A^{\Delta}\to (H,V).$$
Then $$\mathsf{EF}\phi$$
is recognized by  a homomorphism
$$\beta: A^{\Delta}\to (H,V)\circ {\cal U}_1,$$
where $\pi\beta=\alpha.$ 

\noindent{(b)} Suppose that $L\subseteq H_A$ is recognized by a homomorphism
$$\beta:A^{\Delta}\to (H,V)\circ {\cal U}_1.$$
Then $L$ is a boolean combination of languages of the form $\mathsf{EF}(a\wedge\phi),$ where $L_{\phi}$ is recognized by $\pi\beta.$
\end{prop}

\medskip

Here we prove an analogous result for the temporal operator $\mathsf{EX}.$

\medskip\begin{prop}\label{prop.exoperator}
\noindent{(a)}  Suppose that $\phi$ is a tree formula, $\Psi$ is the set of forest formulas of $\phi,$ and that $\Psi$ is recognized by 
$$\alpha:A^{\Delta}\to (H,V).$$
Then $$\mathsf{EX}\phi$$
is recognized by  a homomorphism
$$\alpha\otimes\beta: A^{\Delta}\to (H,V)\circ (H',V'),$$
where $\beta:(A\times H)^{\Delta}\to (H',V')$ is 1-definite.
 \medskip

\noindent{(b)} Suppose that $L\subseteq H_A$ is recognized by a homomorphism
$$\alpha\otimes\beta:A^{\Delta}\to (H,V)\circ (H',V'),$$
Suppose further that every language recognized by $\alpha$ is defined by a formula in some set $\Psi$ of formulas.
 If 
 $$\beta:(A\times H)^{\Delta}\to (H',V')$$ is 1-definite, then $L$ is a boolean combination of languages of the form $L_{\psi}$ and $\mathsf{EX}(a\wedge\psi),$ where $\psi\in\Psi.$
\end{prop}
\begin{proof} 
\noindent{\it (a)} Let $\phi$ be a tree formula. Again, we can write $\phi$ as a disjunction of formulas of the form $a\wedge\psi,$ where $\psi$ is a forest formula of $\phi.$  Since $\mathsf{EX}$ commutes with disjunction, we can write
$$\mathsf{EX}\phi=\bigvee_{j=1}^m\mathsf{EX}(a_j\wedge \psi_j),$$
where each $\psi_j$ is a forest formula of $\phi.$  We now set
$$\beta=\alpha_{A\times H,1}: (A\times H)^{\Delta}/{\sim}_1,$$
  It follows that if $s=a'_1s_1+\cdots a'_ns_n,$ then
$$\alpha\otimes\beta(s)=(\alpha(s),\{(a'_i,\alpha(s_i)):1\leq i\leq n\}).$$
Thus $s\models \mathsf{EX}\phi$ if and only if the second component of $\alpha\otimes\beta(s)$ contains a pair $(a',h),$ where for some $1\leq j\leq m,$ $a'=a_j,$ and forests mapping to $h$ under $\alpha$ satisfy $\psi_j.$  So $\alpha\otimes\beta$ recognizes $\mathsf{EX}\phi$.

\noindent{\it (b)} For the converse, suppose $L$ is recognized by a homomorphism $\alpha\otimes\beta$ as described.  Let $s\in H_A,$ and write $s=a_1s_1+\cdots +a_ns_n.$  Then
$\alpha\otimes\beta(s)=(\alpha(s),h').$  By 1-definiteness of $\beta,$ the value of $h'$ is completely determined by the set of pairs $\{(a_i,\alpha(s_i)):1\leq i\leq n\}.$ The first component of $\alpha\otimes\beta(s)$ is $h\in H,$  if and only if $s\in \alpha^{-1}(h)$ which is defined by a formula $\psi_h\in\Psi.$  The second component of  contains the element $(a,h)$ if and only if $s\models \mathsf{EX}(a\wedge\psi_h).$  Thus $L$ is a boolean combination of sets of the required form.
\end{proof}
\section{$\mathsf{EF}$-algebras}\label{sec.ef}

Following ~\cite{BSW}, we define:
\begin{defn} A finite forest algebra $(H,V)$ is an $\mathsf{EF}$-algebra if it satisfies the identities
$$h+h'=h'+h,  vh+h=vh$$
for all $h,h'\in H$ and $v\in V.$  The second identity with $v=1$ gives $h+h=h.$  Thus every $\mathsf{EF}$-algebra is horizontally idempotent and commutative.
\end{defn}

The following result is proved in ~\cite{BSW}, and is the key element in the characterization of languages definable in one of the temporal logics we consider in   Section~\ref{section.logic}.  We will give a new proof, as it provides a good first illustration of how we use the reachability ideal theory introduced above in decomposition arguments.

\medskip\begin{thm}\label{thm.efwreath}
Let 
$$\alpha:A^{\Delta}\to (H,V)$$
be a homomorphism onto a forest algebra.  $(H,V)$ is an $\mathsf{EF}$-algebra if and only if $\alpha$ factors through a homomorphism
$$\beta:A^{\Delta}\to{\cal U}_1\circ\cdots\circ{\cal U}_1.$$
\end{thm}
\begin{proof} For the `if' direction, suppose $\alpha:A^{\Delta}\to (H,V)$ factors through such a homomorphism $\beta.$ Then $(H,V)$ divides ({\it i.e.,} is a quotient of a subalgebra of) the iterated wreath product.  Since it is obvious that the identities defining $\mathsf{EF}$-algebras are preserved under subalgebras and quotients, and are satisfied by ${\cal U}_1,$ we only need to prove that the wreath product of two $\mathsf{EF}$-algebras is an $\mathsf{EF}$-algebra. 

Clearly idempotence and commutativity of the horizontal monoid are preserved by the wreath product, since we are just forming the direct product  of the component horizontal monoids.  It remains to show that the identity $vh+h=vh$ is preserved by the wreath product.  We have
\begin{eqnarray*}
(v,f)(h_1,h_2)+(h_1,h_2) &=& (vh_1,f(h_1)h_2)+(h_1,h_2)\\
&=& (vh_1+h_1,f(h_1)h_2+h_2)\\
&=& (vh_1, f(h_1)h_2)\\
&=& (v,f)(h_1,h_2).
\end{eqnarray*}

We now prove the converse.  We suppose that $(H,V)$ is an $\mathsf{EF}$-algebra, and show by induction on $|H|$ that any homomorphism into $(H,V)$ factors through such an iterated wreath product. If $|H|=2$ then the $\mathsf{EF}$ identities force $H=\{0,\infty\},$ and either $|V|=1,$ or $(H,V)={\cal U}_1.$  So we may suppose $|H|>2.$  $H$ has trivial reachability classes, because $h\leq h'$ implies $h=vh'=vh'+h'=h+h',$ and similarly $h'\leq h$ implies $h'=h+h',$ so elements in the same reachability class are all equal. In particular $\Gamma_{\mathsf{min}}$ has a single element, so $(H_{\Gamma_{\mathsf{min}}},V_{\Gamma_{\mathsf{min}}})=(H,V).$  Thus by Lemma~\ref{lemma.subminimal}, any homomorphism onto $(H,V)$ factors through a direct product, and hence a wreath product, of $\mathsf{EF}$-algebras with a unique subminimal reachability class. So it suffices to prove the theorem in the case where $(H,V)$ has a single trivial subminimal class $\Gamma=\{h^*\}.$
We claim that in this case
$\alpha:A^{\Delta}\to (H,V)$
factors through
$$\gamma=\alpha_{\Gamma}\otimes\beta:A^{\Delta}\to (H_{\Gamma},V_{\Gamma})\circ {\cal U}_1$$
for some homomorphism $\beta:(A\times H_{\Gamma})^{\Delta}\to {\cal U}_1.$ Since $|H_{\Gamma}|=|H|-1,$ the
desired result follows from the inductive hypothesis. 

We define the homomorphism $\beta$ by giving the value of $\beta(a,h)\in\{1,0\}$ for every $(a,h)\in A\times H_{\Gamma}.$
If $h>\Gamma$ we set $\beta(a,h)=0$ if  $\alpha(a)\cdot h=\infty$ in $(H,V)$ and $\beta(a,h)=1$ otherwise.  If $h$ is the minimal element of $H_{\Gamma},$ we set $\beta(a,h)=0$ if $a\cdot h^*=\infty,$ and $\beta(a,h)=1$ otherwise.

We first establish the following fact:  $\alpha(s)=\infty\in H$ if and only if there is a node $x$ in $s$ such that the tree rooted at $x$ is $at,$ with $\beta(a,\alpha_{\Gamma}(t))=0.$ It follows trivially from the definition that if such a node exists then $\alpha(at)=\infty\in H,$ and thus $\alpha(s)=\infty,$ since $\infty$ is the unique minimal element of $H.$ Conversely, suppose that $\alpha(s)=\infty.$ There must be some node $x$ such that the tree $at$ rooted at $x$ has $\alpha(at)=\infty$:  the alternative would be  that $s=t_1+\cdots + t_r,$ where each $t_i$ is a tree with $\alpha(t_i)\neq \infty$ but $\alpha(s)=\infty.$  This cannot occur, because then by the uniqueness of the subminimal element, each $\alpha(t_i)\geq h^*,$ and the identities for $\mathsf{EF}$-algebras would then give
$$h^*=h^*+\alpha(t_1)+\cdots+\alpha(t_r)=h^*+\alpha(s)+\infty=\infty,$$
a contradiction.  We thus choose a node of maximal depth such that the tree $at$ rooted at this node has $\alpha(at)=\infty.$  By the maximal depth condition, no tree component of $t$ is mapped by $\alpha$ to $\infty,$ and by the argument we just gave $\alpha(t)\neq\infty.$  So $\alpha(t)\geq h^*,$ and thus $\beta(a,\alpha_{\Gamma}(t))=0.$

We now have
$$\gamma(s)=(\alpha_{\Gamma}(s),\beta(s^{\alpha_{\Gamma}})).$$
The left-hand component of $\gamma(s)$ determines $\alpha(s)$ except for distinguishing between $\alpha(s)=h^*$ and $\alpha(s)=\infty.$ The fact that we just proved shows that the right-hand component of $\gamma(s)$ is $\infty$ if and only if $\alpha(s)=\infty.$ Thus $\gamma(s)$ completely determines $\alpha(s).$
\end{proof}

A classic result of Stiffler~\cite{STI} shows that a right transformation monoid $(Q,M)$ divides an iterated wreath product of  copies of the transformation monoid $U_1=(\{0,1\},\{0,1\})$ if and only if $M$ is ${\cal R}$-trivial.  In terms of transformation monoids this means there is no pair of distinct states $q\neq q'\in Q$ such that $qm=q',q'm'=q$ for some $m,m'\in M.$  Since forest algebras are left transformation monoids, the analogous result would suggest that a forest algebra $(H,V)$ divides an iterated wreath product of copies of ${\cal U}_1$ if and only if $V$ is ${\cal L}$-trivial---that is, if and only if $(H,V)$ has trivial reachability classes.  We have already seen that this condition is necessary.

However, the following example shows that it is not sufficient.
\begin{example}\label{example.mainexample}
Figure~\ref{fig.counterexample} below defines the syntactic forest algebra of the language $L$ of Example~\ref{example.efexformula}. The nodes in the diagram represent the elements of the horizontal monoid, and the arrows give the action of a generating set of letters $A=\{a,b\}$ on the horizontal monoid. The letter transitions, together with the conventions about idempotence and commutativity, and the meaning of 0 and $\infty,$ completely determine the addition and the action.
 \begin{figure}[htb]\label{figure.mainexample}
       \begin{center}\vspace*{-2mm}
            \begin{tikzpicture}[->,node distance=1.5cm,semithick]
	      \tikzset{every state/.style={inner sep=0pt}}
              \node[state] (0)              {$0$};
              \node[state] (A) [right of=0] {$a$};
              \node[state] (B) [right of=A] {$b$};
              \node[state] (C) [right of=B] {$a\!+\!b$};
            
              \path (0) edge             node[above] {$a$}   (A)
		    (A)	edge[loop above] node        {$a$}   (A)
		    (A) edge		 node[above] {$b$}   (B)
		    (B) edge[loop above] node        {$a,b$} (B)
		    (0) edge[bend right] node[below] {$b$}   (B)
		    (C) edge[loop above] node        {$a,b$} (C);
            \end{tikzpicture}\vspace*{-3mm}
	\caption{An algebra with trivial reachability classes that is not an $\mathsf{EF}$-algebra}\label{fig.counterexample}
      \end{center}
     \end{figure}
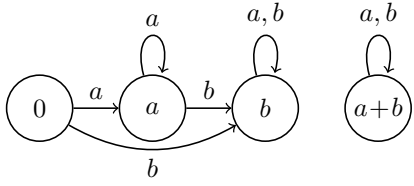 
Since $\infty = a+ b= a+ba\neq ba = b,$ this is not an $\mathsf{EF}$-algebra, but the reachability classes are singletons.
\end{example}

\section{Definiteness}\label{sec.ex}
\subsection{Definite homomorphisms}
Let $k>0.$  A finite semigroup $S$ is said to be {\it reverse $k$-definite} if it satisfies the identity
$$x_1x_2\cdots x_ky=x_1\cdots x_k.$$
The reason for the word `reverse' is that definiteness of semigroups was originally formulated in terms of right transformation monoids, so the natural analogue of definiteness in the setting of forest algebras corresponds to reverse definiteness in semigroups.  Observe that the notions of definiteness and reverse definiteness in semigroups do not really make sense for monoids, since only the trivial monoid can satisfy the underlying identities.
For much the same reason, we define definiteness for forest algebras not as a property of the algebras themselves, but of homomorphisms
 $$\alpha:A^{\bf \Delta}\to (H,V).$$
 
 The {\it depth} of a context $p\in V_A$ is defined to be the depth of its hole; so for instance a context with its hole at a root node has depth 0. We say that the homomorphism $\alpha$ is $k$-definite, where $k>0,$ if for every $p\in V_A$ of depth at least $k,$ and for all $s,s'\in H_A,$ $\alpha(ps)=\alpha(ps').$    Easily, if $\alpha_1,\alpha_2$ are $k$-definite homomorphisms, then so are $\alpha_1\times\alpha_2$ and $\psi\alpha_1,$ where $\psi:(H,V)\to (H',V')$ is a homomorphism of forest algebras.
 
A context is {\it guarded} if it has depth at least 1, that is, if the hole is not at the root.  We denote by $V_A^{\mathsf{gu}}$ the subsemigroup of $V_A$ consisting of the guarded contexts.
\medskip\begin{lemma}\label{lemma.definite} Let $k>0.$ A homomorphism $\alpha:A^{\bf \Delta}\to (H,V)$ is $k$-definite if and only if
$\alpha(V_A^{\mathsf{gu}})$ is a reverse $k$-definite semigroup.
\end{lemma}
\begin{proof} Let $\alpha$ be $k$-definite, and let $p_1,\ldots, p_k\in V_A^{\mathsf{gu}}.$  Then the context $p=p_1\cdots p_k$ has depth at least $k.$ Thus if $q\in V_A^{\mathsf{gu}},$ and $s\in H_A,$ we have
$$\alpha(p_1\cdots p_kqs)=\alpha(pqs)=\alpha(ps)=\alpha(p_1\cdots p_ks).$$
As this holds for arbitrary $s,$ faithfulness of the action implies
$$\alpha(p_1\cdots p_kq)=\alpha(p_1\cdots p_k),$$
and thus $\alpha(V_A^{\mathsf{gu}})$ is a reverse $k$-definite semigroup.

Conversely, suppose that $\alpha(V_A^{\mathsf{gu}})$ is a reverse $k$-definite semigroup.  Let $p$ be a context of depth at least $k.$  By following the path from the hole of $p$ to a root, we obtain a factorization $p=q_1\cdots q_r,$ where $r\geq k,$ and each $q_i$ has the form $t+a\Box + t',$ where $a\in A$ and $t,t'\in H_A.$  In particular, we can write $p=p_1\cdots p_k,$ where each $p_i\in V_A^{\mathsf{gu}}.$  Let $s\in H_A,$ with $s\neq 0.$  Then we can write $s=q\cdot 0,$ where $q\in V_A^{\mathsf{gu}}.$  Thus
$$\alpha(ps)=\alpha(p_1\cdots p_kq\cdot 0)=\alpha(p_1\cdots p_k\cdot 0)=\alpha(p\cdot 0).$$
As this holds for arbitrary $s,$ $\alpha$ is $k$-definite.
\end{proof}

\medskip

\begin{defn} An $\mathsf{EX}$-homomorphism is a homomorphism that is $k$-definite for some $k\in\mathbb N$.
\end{defn}

\subsection{Free $k$-definite algebra}
We construct what we will call {\it free $k$-definite algebra} over an alphabet $A.$ This is a slight abuse of terminology, since as we noted above, it is the homomorphism into this algebra, and not the algebra itself, that is $k$-definite. We do this by recursively defining a sequence of congruences $\sim_k$ on $A^{\Delta}.$  If $k=0,$ then $\sim_0$ is just the trivial congruence that identifies all forests.  If  $k\geq 0$ and $\sim_k$ ha been defined then we associate to each forest $s=a_1s_1+\cdots a_rs_r,$ where each $a_i\in A,$ $s_i\in H_A,$ the set
$$T_s^{k+1}=\{(a_i,[s_i]_{\sim_k}):1\leq i\leq r\},$$
where $[]_{\sim_k}$ denotes the $\sim_k$-class of a forest.  We then define $s\sim_{k+1}s'$ if and only if $T_s^{k+1}=T_{s'}^{k+1}.$
\medskip\begin{prop}\label{prop.freedefinite}  Let $k\geq 0.$ Then 
$\sim_{k+1}$ refines $\sim_k.$
$\sim_k$ is a congruence of finite index on $A^{\Delta},$ with a horizontally idempotent and commutative quotient.
 \end{prop}
\begin{proof} Obviously $\sim_k$ is an equivalence relation of finite index.  We prove by induction on $k$ that it is also a congruence with a horizontally idempotent and commutative quotient.  The case $k=0$ is obvious.  Assume now that $\sim_k$ is a congruence with an idempotent and commutative quotient.  If $s\sim_{k+1} s',$ then idempotence and commutativity of $\sim_k$ gives
 $$s\sim_k \sum_{(a,[t])\in T_s^{k+1}}at\sim_k s',$$ which proves the first claim.  We further have
$$T_{as}^{k+1}=\{(a,[s]_{\sim_k})\}=\{(a,[s']_{\sim_k})\}=T_{as'}^{k+1},$$ 
so $as\sim_{k+1}as'.$  Moreover, if $s_i\sim_{k+1}s_i'$ for $i=1,2,$ then $s_1+s_2\sim_{k+1}s_1'+s_2',$  since addition of forests corresponds to union of the associated sets. Since the equivalence is preserved under application of a letter and under addition, it is a congruence.  The observation about addition and union implies that the quotient is idempotent and commutative.
\end{proof}

Intuitively, $s\sim_k s'$ means that the forests $s$ and $s'$ are identical at the $k$ levels closest to the root, up to idempotent and commutative equivalence.  In fact, this intuition provides an equivalent characterization of $\sim_k$, which we give below.  We omit the simple proof.

\medskip\begin{lemma}\label{lemma.icdef} Let $s, s'\in H_A$ and $k>0.$  Let $\bar s,$ $\bar{s'},$ denote, respectively, the forests obtained from $s$ and $s'$ by removing all the nodes at depth $k$ or more.  Then $s\sim_k s'$ if and only if $\bar s$ and $\bar{s'}$ are idempotent-and-commutative equivalent.
\end{lemma}

Let us denote by $\alpha_{A,k}$ the homomorphism from $A^{\Delta}$ onto its quotient by $\sim_k.$  In the case where $k=1,$ we will identify $H_A/{\sim}_1$ with the monoid $({\cal P}(A),\cup),$ and the horizontal component of $\alpha_{A,1}$ with the map that sends each forest to the set of its root nodes.

The following theorem gives both the precise sense in which this is the `free $k$-definite forest algebra', as well as the wreath product decomposition of $k$-definite homomorphisms into 1-definite homomorphisms into a forest algebra with horizontal monoid $\{0,\infty\}.$

\medskip\begin{thm}\label{thm.def} Let $\alpha:A^{\Delta} \to (H,V)$ be a homomorphism onto a finite forest algebra.  Let $k>0.$  The following are equivalent.
\begin{itemize}
\item[\it (a)] $\alpha$ is $k$-definite.
\item[\it (b)] $\alpha$ factors through $\alpha_{A,k}.$
\item[\it (c)] $\alpha$ factors through 
$$\beta_1\otimes\cdots\otimes\beta_k:A^{\Delta} \to(H_1,V_1)\circ\cdots\circ (H_k,V_k),$$
where each $$\beta_i:(A\times H_1\times\cdots \times H_{i-1})^{\Delta}\to (H_i,V_i)$$
is 1-definite.
\item[\it (d)] $\alpha$ factors through an iterated wreath product of  1-definite homomorphisms into ${\cal U}_2.$
\end{itemize}
\end{thm}
\begin{proof}
\noindent $({\it (a)}\Rightarrow {\it (b)})$. 
Let $s\sim_k s'$ and let $\bar{s},$ $\bar{s'}$ denote the forests obtained by removing all nodes at depth $k$ or more from $s$ and $s'.$  Let $\alpha:A^{\Delta}\to (H,V)$ be $k$-definite. By Lemma~\ref{lemma.icdef}, $\bar s$ and $\bar{s'}$ are idempotent-and-commutative equivalent. Thus by Lemma~\ref{lemma.ic_equivalence}, $\alpha(\bar s)=\alpha(\bar{s'}).$ If $t_2\in H_A$ is obtained from $t_1$ by removing a node of depth $k$ together with all its descendants, then there is a context $p$ with a hole at depth $k$ and a forest $u$ such that $t_1=pu,$ $t_2=p0$  Thus by $k$-definiteness, $\alpha(t_2)=\alpha(t_1).$  Consequently for any forest $s,$ $\alpha(\bar{s})=\alpha(s).$  Thus with $s\sim_k s'$ as above, we have 
$$\alpha(s)=\alpha(\bar{s})=\alpha(\bar{s'})=\alpha(s').$$ 
So $\alpha$ factors through $\alpha_{A,k},$ as required.

\medskip

\noindent $({\it (b)}\Rightarrow {\it(c)})$. It suffices to show that for $k\geq 1,$  $\alpha_{A,k+1}$ factors through a homomorphism $\gamma=\alpha_{A,k}\otimes \beta,$ where $\beta$ is 1-definite.  In fact, we can choose $\beta=\alpha_{A\times H_A/{\sim}_k,1}.$  We then have, for $s\in H_A,$
\begin{eqnarray*}
\gamma(s)&=& ([s]_{\sim_k},\beta(s^{\alpha_{A,k}}))\\
&=& ([s]_{\sim_k},T_s^{k+1})
\end{eqnarray*}
so that in fact $s\sim_{k+1}s'$ if and only if $\gamma(s)=\gamma(s').$  (That is, the two homomorphisms factor through one another, and so are essentially identical.)
 
 \medskip
 
 \noindent $({\it (c)}\Rightarrow {\it (d)})$  It will suffice to show that for any finite alphabet $A,$ the homomorphism $\alpha_{A,1}$ factors through a wreath product of 1-definite homomorphisms into ${\cal U}_2.$  In fact, we will show that $\alpha_{A,1}$ is isomorphic to a direct product of such homomorphisms.  For each $a\in A$ define
 $$\beta_a:A^{\Delta}\to {\cal U}_2$$
 by setting $\beta_a(a)=0,$ $\beta_a(b)=c_0$ for $b\in A-\{a\}.$   If $s\in H_A,$ then $\beta_a(s)=\infty$ if some root node of $s$ is $a,$ and $\beta_a(s)=0$ otherwise.  In particular, $\beta_a(s)$ depends only on the set of labels of the root nodes of $s$ and so is 1-definite.  Consider the direct product
 $$\beta=\biggl(\prod_{a\in A}\beta_a\biggr):A^{\Delta}\to\prod_{a\in A}{\cal U}_2.$$
 Then $\beta(s)$ is an $A$-tuple from $\{0,\infty\}$ in which the components  with value $\infty$ are exactly those corresponding to the labels of the root nodes of $s.$  Since $\alpha_{A,1}(s)$ is the set of labels of root nodes of $s,$ the two homomorphisms are equivalent.
  
 \medskip

 \noindent $({\it (d)}\Rightarrow {\it(a)})$. It suffices to show that for homomorphisms
 $$\alpha:A^{\Delta}\to (H_1,V_1), \beta: (A\times H_1)^{\Delta}\to (H_2,V_2),$$
 where $\alpha$ is $k$-definite and $\beta$ is 1-definite, that $\gamma=\alpha\otimes\beta$ is $(k+1)$-definite.
 That is, we will show that if $p$ is a context with a hole at depth $k+1$ and $s$ is a forest, then $\gamma(ps)$ is
 independent of $s.$  We can write $ps=u+aqs+v,$ where $u,v\in H_A,$ $a\in A,$ and $q$ is a context with a
 hole at depth $k.$  It is thus sufficient to show that $\gamma(aqs)$ is independent of $s.$  But we have
 \begin{eqnarray*}
 \gamma(aqs) &=& \gamma(a)\gamma(qs)\\
 &=&\gamma(a)(\alpha(qs),\beta((qs)^{\alpha}))\\
 &=& (\alpha(a)\cdot\alpha(qs),\beta(a,\alpha(qs))\cdot \beta((qs)^{\alpha})).
 \end{eqnarray*}
 Since $\alpha$ is $k$-definite, $\alpha(qs)$ depends only on $q.$  
 Since $\beta$ is 1-definite, the right-hand coordinate depends only on $\beta(a,\alpha(qs)),$
 which depends only on $a$ and $q.$  Thus the value is independent of $s,$ as required.
\end{proof}

\section{$(\mathsf{EF,EX})$-algebras}\label{sec.efex}
\subsection{The principal result}

\begin{defn}
An {\it $(\mathsf{EF,EX})$-homomorphism} $\alpha:A^{\Delta}\to (H,V)$ is one that factors through an iterated wreath product
$$\beta_1\otimes\cdots\otimes\beta_k,$$
where each $\beta_i$ either maps into ${\cal U}_1$ or is 1-definite.  By Theorem~\ref{thm.def} we can suppose that each 1-definite $\beta_i$ maps into ${\cal U}_2.$
\end{defn}

The principal result of this paper is an effective necessary and sufficient condition for a homomorphism to be a   $(\mathsf{EF,EX})$-homomorphism.  
\medskip\begin{defn}
Suppose $\alpha:A^{\Delta}\to (H,V)$. Let $s_1, s_2\in H_A,$ $k>0,$ and $\Gamma\subseteq H$ a reachability class for $(H,V).$  We say that $s_1, s_2$ are {\it $(\alpha,k,\Gamma)$-confused,} and write $s_1\equiv_{\alpha,k,\Gamma}s_2,$ if 
$$(s_1)^{\alpha_{\Gamma}}\sim_k (s_2)^{\alpha_{\Gamma}},\;\;\;\alpha(s_1),\alpha(s_2)\in \Gamma.$$
\end{defn}
Observe that the equivalence relation $\sim_k$  in the first item is over the extended alphabet $A\times{H_{\Gamma}}.$  It is worth emphasizing what $(s)^{\alpha_{\Gamma}}$ is when $\alpha(s)\in\Gamma$: We are tagging each node of $x$ of $s$ with the value $\alpha(t)\in H$ if the tree rooted at $x$ is $at$ and $\alpha(t)>\Gamma,$ but we are tagging the node by $\infty$--effectively leaving the node untagged--if $\alpha(t)\in\Gamma.$ Since $\alpha(s)\in\Gamma,$ every node is of one of these two types.

\medskip\begin{defn}
A homomorphism $\alpha$ is {\it nonconfusing} if and only if there exists $k>0$ such that $\equiv_{\alpha,k,\Gamma}$ is equality for reachability classes $\Gamma.$ 
\end{defn}

It follows from Proposition~\ref{prop.freedefinite} that $\equiv_{\alpha,k+1,\Gamma}$ refines $\equiv_{\alpha,k,\Gamma},$  so that if $\alpha$ is nonconfusing with associated parameter $k,$ then it is nonconfusing for all $m>k.$

Our main result is:
 \medskip\begin{thm}\label{thm.main} Let $\alpha:A^{\Delta}\to (H,V)$ be a homomorphism into a finite forest algebra.  $\alpha$ is a $(\mathsf{EF,EX})$ homomorphism if and only if it is nonconfusing.  \end{thm}

The proof of Theorem~\ref{thm.main} will be given in the next two subsections.  
\begin{example}
Consider once again the algebra  of Examples~\ref{example.efexformula} and \ref{example.mainexample} and the associated homomorphism  $\alpha$ from $\{a,b\}^{\Delta}$.  Since the algebra has trivial reachability classes, $\alpha$ is nonconfusing for all $k$, so Theorem~\ref{thm.main} implies that $\alpha$ is  an  $(\mathsf{EF,EX})$-homomorphism.  We will see in the course of the proof of the main theorem how the wreath product decomposition is obtained.
\end{example}
\begin{example} Consider again the forest algebra
${\cal U}_2 = (\{0,\infty\},\{1,c_{\infty},c_0\}),$
and the homomorphism $\alpha$ from $\{a,b,c\}^{\Delta}$ onto ${\cal U}_2$ that maps $a$ to 1, $b$ to $c_0$ and $c$ to $c_{\infty}$.  There is a unique reachability class $\Gamma,$ so for any forest $s,$ $s^{\alpha_{\Gamma}}$ is identical to $s.$  Now observe that $a^kb\sim_k a^kc,$ but that these are mapped to different elements under $\alpha.$  So by our main theorem, $\alpha$ is not an  $(\mathsf{EF,EX})$-homomorphism.%
\end{example}

\subsection{Sufficiency of the condition}
We will use the ideal theory developed earlier to prove that every nonconfusing homomorphism factors through  a wreath product decomposition of the required kind.  The structure of our proof resembles the one given for Theorem~\ref{thm.efwreath} Once again, we proceed by induction on $|H|.$  The base of the induction  is the trivial case $|H|=1.$ Let us suppose that
$$\alpha:A^{\Delta}\to (H,V)$$
is nonconfusing with parameter $k,$  that $|H|>1,$ and that every nonconfusing homomorphism into a forest algebra with a smaller horizontal monoid factors through a wreath product of the required kind.  

Let $\Gamma=\Gamma_{\mathsf{min}}.$ Suppose first that $|\Gamma|>1.$ We claim that $\alpha$ factors through
$$\beta=\alpha_{\Gamma}\otimes\alpha_{B,k}: A^{\Delta}\to (H_{\Gamma},V_{\Gamma})\circ B^{\Delta}/{\sim}_k$$
where $B=A\times H_{\Gamma}.$  Since $|H_{\Gamma}|<|H|$ and $\alpha_{\Gamma}$ is also nonconfusing, the induction hypothesis gives the desired decomposition of $\alpha.$ To establish the claim, let $s\in H_A.$  Then
$$\beta(s)=(\alpha_{\Gamma}(s),[s^{\alpha_{\Gamma}}]_{\sim_k}).$$
If $s\notin\Gamma,$ then the value of the left-hand coordinate determines $\alpha(s).$  If $s\in\Gamma,$ then by the nonconfusion condition, the value of the right-hand coordinate determines $\alpha(s).$  Thus $\alpha$ factors through $\beta$ as required. 

So let 
$|\Gamma|=1.$  Then $\Gamma=\{\infty\}$ and  $(H_{\Gamma},V_{\Gamma})=(H,V).$  Lemma~\ref{lemma.subminimal} implies that we can suppose $(H,V)$ has a single subminimal reachability class, because each of the component homomorphisms in the direct product is nonconfusing, and the direct product factors through the wreath product.

Thus we have a unique minimal element $\infty,$ and a unique subminimal ideal $\Gamma'.$  We claim that $\alpha$ factors through
$$\beta=\alpha_1\otimes\alpha_2\otimes\alpha_3:A^{\Delta}\to(H_{\Gamma'},V_{\Gamma'})\circ B^{\Delta}/{\sim}_k\circ\;{\cal U}_1,$$
where 
\begin{itemize}
\item[] $\alpha_1=\alpha_{\Gamma'}.$
\item[] $\alpha_2 = \alpha_{B,k},$ where $B=A\times H_{\Gamma'}.$
\item[] $\alpha_3:(B\times 2^{B})^{\Delta}\to {\cal U}_1$ will be defined below.
\end{itemize}
To see how $\alpha_3$ should be defined, let us consider what this homomorphism needs to tell us.  If $\alpha(s)>\Gamma',$ then the first coordinate of $\beta(s)$ determines $\alpha(s).$  If $\alpha(s)\in\Gamma',$ then the first two components of $\beta(s)$ determine $\alpha(s),$ by nonconfusion.  So we will use the third component to distinguish between $\alpha(s)\in\Gamma'$ and $\alpha(s)=\infty.$  The value of the first component already determines whether or not $\alpha(s)\in\Gamma'\cup\{\infty\},$ so we really just need to be able to tell when $\alpha(s)=\infty.$  There are several cases to consider, depending on whether or not $s$ contains a tree $t$ such that $\alpha(t)=\infty.$  If not, then $s=t_1+\cdots+t_r,$ where $\alpha(t_i)\geq \Gamma'$ for all $i.$ Observe that if this is the case, then the set of values $\{\alpha(t_1),\cdots,\alpha(t_r)\}$ is determined by the second component $\{[t_1^{\alpha_1}]_{\sim_k},\ldots,[t_r^{\alpha_1}]_{\sim_k}\}$ of $\beta(s).$  If $s$ contains a tree $t$ such that $\alpha(t)=\infty,$ pick such a tree at maximal depth.  Then $t= a(t_1+\cdots+t_r),$ where once again $\alpha(t_i)\geq\Gamma'$ for all $i,$ and the set of values $\{\alpha(t_1),\cdots,\alpha(t_r)\}$ is determined by the second component of $\beta(s).$
We now specify the value of $\alpha_3(a,h,Q).$ As remarked above, $Q$ determines  a set of values all in $\Gamma'$ or strictly higher.  Let $h_Q\in H_A$ be the sum of these values.  If either $h_Q=\infty,$ or $ah_Q=\infty,$ set $\alpha_3(a,h,Q)=0.$  Otherwise, $\alpha_3(a,h,q)=1.$  

The third component of $\beta(s)$  will be $\infty$ if and only if there is some subtree $a(t_1+\cdots + t_r)$ such that
$$\alpha_3(a,\alpha_1(t_1+\cdots+t_r),\{[t_1^{\alpha_1}]_{\sim_k},\ldots,[t_r^{\alpha_1}]_{\sim_k}\})=0.$$
If we pick the subtree of maximal depth at which this occurs, then as argued above, $\alpha(s)=\infty.$  The only other way we can have $\alpha(s)=\infty$ is if there is no such subtree, but $s=t_1+\cdots +t_r$ where each $\alpha(t_i)\geq\Gamma'$ and the sum of these values is $\infty.$  In this case, the fact that no such subtree exists is determined by the third coordinate of $\beta(s)$ being 1, and the set of $\alpha(t_i)\geq\Gamma'$ is determined by the second coordinate of $\beta(s).$  So in all cases $\beta(s)$ determines $\alpha(s).$

\subsection{Necessity of the condition.}
To prove the converse, we have to show preservation of nonconfusion under quotients and wreath products with the allowable factors.  This is carried out in the following three lemmas.  Preservation under quotients (Lemma~\ref{lemma.quotientclosure}) is the most difficult of  the three to show.

\medskip\begin{lemma}\label{lemma.quotientclosure}
Let 
$$\alpha:A^{\Delta}\to (H_1,V_1), \beta:A^{\Delta}\to (H_2,V_2),$$
be homomorphisms onto finite forest algebras such that $\beta$ factors through $\alpha$.  If $\alpha$ is nonconfusing then so is $\beta.$

\end{lemma}
\begin{proof}
Suppose that $\alpha$ is nonconfusing with parameter $k.$  We will show that $\beta$ is nonconfusing with the same parameter.  To this end, let $\Gamma\subset H_2$ be a reachability class, and let $s_1, s_2\in H_A$ be forests with 
$$s_1^{\beta_{\Gamma}}\sim_k s_2^{\beta_{\Gamma}}
\mbox{ and }\beta(s_1), \beta(s_2)\in\Gamma.$$
We must show $\beta(s_1)=\beta(s_2).$

Since $\beta$ factors through $\alpha,$ there is an onto forest algebra homomorphism $\eta:(H_1,V_1)\to (H_2,V_2)$ such that $\beta=\eta\alpha.$
Choose $h_2\in\Gamma,$ and let $h_1\in H_1$ be such that $\eta(h_1)=h_2,$ and $h_1$ is $\leq$-minimal for this property.  Let $\Lambda$ be the reachability class of $h_1.$  By Lemma~\ref{lemma.morphisms}, $\eta(\Lambda)=\Gamma.$

  We now perform a little surgery on the forests $s_1$ and $s_2$: For each $h\in H_2,$ we choose $s_h\in H_A$ such that $\beta(s_h)=h,$ and such that if $h\in \Gamma,$ then $\alpha(s_h)\in\Lambda.$  (This is where we use the fact that $\eta(\Lambda)=\Gamma.$) Look at a nodes of $s_1$ or $s_2$ at depth $k-1.$  If the tree rooted at such a node is $at,$ where $t\in H_A,$ we replace the forest $t$ by $s_{\beta(t)}.$  We do this for every node at depth $k-1$ of the two forests and obtain new forests $\bar s_1$ and $\bar s_2.$ Obviously we have not changed the values of these forests under $\beta,$ so we have
$\beta(\bar s_i)=\beta(s_i)$ for $i=1,2.$  Thus it is sufficient to show $\beta(\bar s_1)=\beta(\bar s_2).$

We claim that 
$$(\bar s_1)^{\alpha_{\Lambda}}\sim_k (\bar s_2)^{\alpha_{\Lambda}}
\mbox{ and }
\alpha_{\Lambda}(\bar s_1)=\alpha_{\Lambda}(\bar s_2).$$
This claim gives our desired result.  To see this, note that the second equality in the claim implies $\alpha(\bar s_1)>\Lambda$ if and only if $\alpha(\bar s_2)>\Lambda.$ If $\alpha(\bar s_1),\alpha(\bar s_2)>\Lambda,$  we  have
$$\alpha(\bar s_1)=\alpha_{\Lambda}(\bar s_1)=\alpha_{\Lambda}(\bar s_2)=\alpha(\bar s_2).$$
On the other hand, if $\alpha(\bar s_1),\alpha(\bar s_2)\not >\Lambda,$ we must have $\alpha(\bar s_1),\alpha(\bar s_2)\in\Lambda,$ because $\eta\alpha(\bar s_i)=\beta(\bar s_i)\in\Gamma.$ Then the nonconfusing property of $\alpha$ gives $\alpha(\bar s_1)=\alpha(\bar s_2).$  So in all cases we have $\alpha(\bar s_1)=\alpha(\bar s_2).$  Applying $\eta$ gives $\beta(\bar s_1)=\beta(\bar s_2),$ as required.

We prove the claim by induction on $k.$  More precisely, for each $k\geq 1,$ we will show that if
$s^{\beta_{\Gamma}}\sim_k t^{\beta_{\Gamma}},$
then
$(\bar s)^{\alpha_{\Lambda}}\sim_k (\bar t)^{\alpha_{\Lambda}},$
and
$\alpha_{\Lambda}(\bar s)=\alpha_{\Lambda}(\bar t).$
First suppose $k=1,$ so $s^{\beta_{\Gamma}}\sim_1 t^{\beta_{\Gamma}}.$ A root node of $(\bar s)^{\alpha_{\Lambda}}$ has the form $(a,\alpha_{\Lambda}(s_h)),$ for some $h\in H_2.$  This means that $s$ has a component tree of the form $au,$ where $\beta(u)=h.$ Thus $s^{\beta_{\Gamma}}$ has a root node labeled $(a,h)$ or $(a,\infty),$ depending on whether $h\in\Gamma.$ Since $s^{\beta_{\Gamma}}\sim_1 t^{\beta_{\Gamma}},$ the same root node occurs in $t^{\beta_{\Gamma}}.$ If $h\notin\Gamma,$ then  $t$ contains a component $av$ with $\beta(v)=h,$ and thus $\bar t$ has a component $as_h,$ so that $(\bar t)^{\alpha_{\Lambda}}$ has a root node $(a,\alpha_{\Lambda}(s_h)).$  If $h\in\Gamma,$ then $t$ contains a component $av$ with $\beta(v)=h'\in\Gamma,$ so that $(\bar t)^{\alpha_{\Lambda}}$ contains a root node
$$(a,\alpha_{\Lambda}(s_{h'}))=(a,\infty)=(a,\alpha_{\Lambda}(s_h)),$$
so that every root node of $(\bar s)^{\alpha_{\Lambda}}$ is also a root node of $(\bar t)^{\alpha_{\Lambda}}.$ We get the converse inclusion by symmetry.  So $(\bar s)^{\alpha_{\Lambda}}\sim_1 (\bar t)^{\alpha_{\Lambda}}.$  Moreover,  we also have $\alpha_{\Lambda}(\bar s)=\alpha_{\Lambda}(\bar t).$ This is because if no root node of $(\bar s)^{\alpha_{\Lambda}}$ has the form $(a,\infty),$ then $\alpha_{\Lambda}(\bar s)$ is determined by the sum of the $\alpha(as_h),$ and $\alpha_{\Lambda}(\bar t)$ is determined by the sum of the same set of terms.  If some root node of $(\bar s)^{\alpha_{\Lambda}}$ is $(a,\infty),$ then the same is true for some root node of $(\bar t)^{\alpha_{\Lambda}},$ and we have $\alpha_{\Lambda}(\bar s)=\infty =\alpha_{\Lambda}(\bar t).$

Our induction hypothesis is now that $k\geq 1,$ and that whenever $s^{\beta_{\Gamma}}\sim_k t^{\beta_{\Gamma}},$ we have both $(\bar s)^{\alpha_{\Lambda}}\sim_k (\bar t)^{\alpha_{\Lambda}},$ and $\alpha_{\Lambda}(\bar s)=\alpha_{\Lambda}(\bar t).$  We show that these properties are  preserved at level $k+1.$ 
If we write
$$s=a_1s_1+\cdots +a_ks_k, t=b_1t_1+\cdots + b_pt_p,$$
where the $s_i, t_j$ belong to $H_A,$ and the $a_i,b_j$ to $A,$ then we have
$$\bar s=a_1\bar s_1+\cdots +a_k\bar s_k, t=b_1\bar t_1+\cdots + b_p\bar t_p.$$
It is important to understand precisely what the operator $u\mapsto \bar u$ means in these equations:  On the left-hand sides we are performing the substitution at nodes of $\bar s,\bar t$ at level $k$; on the right-hand sides we carry out the operation at nodes of level $k-1.$  The $\sim_{k+1}$-class of $(\bar s)^{\alpha_{\Lambda}}$ is determined by the set $T^{k+1}_{{\bar s}^{\alpha_{\Lambda}}}$ of pairs of the form
$((a_i,\alpha_{\Lambda}(\bar s_i)),[(\bar s_i)^{\alpha_{\Lambda}}]_{\sim_k}).$
Let $1\leq i\leq r.$  The corresponding set $T^{k+1}_{s^{\beta_{\Gamma}}}$ for the $\sim_{k+1}$-class of $s^{\beta_{\Gamma}}$ contains the pair
$((a_i,\beta_{\Gamma}(s_i)),[s_i^{\beta_{\Gamma}}]_{\sim_k}).$
Thus there is some $j$ such that $a_i=b_j$ and $s_i^{\beta_{\Gamma}}\sim_k t_j^{\beta_{\Gamma}}.$  By the induction hypothesis, we have both $\bar s_i^{\alpha_{\Lambda}}\sim_k \bar t_j^{\alpha_{\Lambda}}$
and $\alpha_{\Lambda}(\bar s_i)=\alpha_{\Lambda}(\bar t_j).$  Thus the pair
$((a_i,\alpha_{\Lambda}(\bar s_i)),[(\bar s_i)^{\alpha_{\Lambda}}]_{\sim_k}),$
also occurs in $T^{k+1}_{{\bar t}^{\alpha_{\Lambda}}}$ This shows $T^{k+1}_{{\bar s}^{\alpha_{\Lambda}}}\subseteq T^{k+1}_{{\bar t}^{\alpha_{\Lambda}}}.$  We get the converse inclusion by symmetry. So
$(\bar s)^{\alpha_{\Lambda}}\sim_{k+1} (\bar t)^{\alpha_{\Lambda}}.$  We obtain $\alpha_{\Lambda}(\bar s)=\alpha_{\Lambda}(\bar t)$ just as we did in the case $k=1.$
\end{proof}

\medskip\begin{lemma}\label{lemma.wreathu1}
Suppose that 
$$\alpha:A^{\Delta}\to  (H,V)\circ {\cal U}_1$$
is a homomorphism, and that $\beta=\pi\alpha,$ where $\pi$ is the projection morphism onto $(H,V),$ is nonconfusing.  Then $\alpha$ is nonconfusing.
\end{lemma}
\begin{proof}
Let $k$ be the nonconfusion parameter for $\beta.$
Let $\Delta\subseteq H\times\{0,\infty\}$ be a reachability class in the image of $\alpha.$ Suppose $s,t\in H_A$ with $(s)^{\alpha_{\Delta}}\sim_k (t)^{\alpha_{\Delta}}$ and $\alpha(s),\alpha(t)\in\Delta.$  We must show $\alpha(s)=\alpha(t).$

We can never reach an element of the form $(h,0)$ from one of the form $(h',\infty),$ since
$$(v,f)(h',\infty)=(vh,f(h)\cdot\infty)=(vh,\infty).$$
So, since $\alpha(s)\cong\alpha(t),$ they must agree in the right-hand coordinate.  It remains to show that the left coordinates $\beta(s),$ $\beta(t)$ are equal.

As we argued in the proof of Lemma~\ref{lemma.quotientclosure}, $\pi(\Delta)$ is contained in a reachability class $\Gamma$ of $H.$  Let us look at the corresponding nodes of $(s)^{\alpha_{\Delta}},$ $(t)^{\alpha_{\Delta}}$ and of  $(s)^{\beta_{\Gamma}}$, $(t)^{\beta_{\Gamma}}$.  If a node of $(s)^{\alpha_{\Delta}}$ has a label of the form $(a,h),$ where $h>\Delta,$ then the corresponding node of $(s)^{\beta_{\Gamma}}$ will either be labeled $(a,\pi(h))$ or $(a,\infty),$ and this is entirely determined by the value of $h.$  If, on the other hand, a node of $(s)^{\alpha_{\Delta}}$has the label $(a,\infty),$ then $h\not >\Delta.$ But since $\alpha(s)\in\Delta,$ this implies $h\in\Delta,$ so that $\pi(h)\in\Gamma,$ and thus the node also has the label $(a,\infty)$ in $(s)^{\beta_{\Gamma}}.$  Thus the labels of nodes of $(s)^{\beta_{\Gamma}}$ and $(t)^{\beta_{\Gamma}},$ are determined by applying a mapping
$H_{\Delta}\to H_{\Gamma}$ to the right coordinates of the node labels of $(s)^{\alpha_{\Delta}},$ $(t)^{\alpha_{\Delta}}$  As a consequence,
$(s)^{\beta_{\Gamma}}\sim_k (t)^{\beta_{\Gamma}}.$  Also, for $i=1,2,$ $\beta(s)=\pi\alpha(s)\in\Gamma,$ and likewise $\beta(t)\in\Gamma,$ so nonconfusion gives $\beta(s)=\beta(t),$ as required.
\end{proof}
\medskip\begin{lemma}\label{lemma.wreath1def}
Suppose that 
$$\alpha=\beta\otimes\gamma:A^{\Delta}\to  (H,V)\circ (H',V')$$
is a homomorphism, that $\beta$ is nonconfusing, and that $\gamma:(A\times H)^{\Delta}\to (H',V')$ is 1-definite.  Then $\alpha$ is nonconfusing.
\end{lemma}
\begin{proof} Let us suppose in particular that $\beta$ is nonconfusing with parameter $k.$ We claim that $\alpha$ is nonconfusing with parameter $k+1.$  Let $\Gamma\subseteq H\times H'$ be a reachability class in the image of $\alpha,$ and let $s,t\in H_A$ with $(s)^{\alpha_{\Gamma}}\sim_{k+1} (t)^{\alpha_{\Gamma}}$ and $\alpha(s),\alpha(t)\in\Gamma.$  We will show $\alpha(s)=\alpha(t).$  We begin by proving that $\beta(s)=\beta(t),$ using the nonconfusing property of $\beta,$ and then use 1-definiteness to show that the right-hand coordinates are also equal.

We write both $s$ and $t$ as sums of the component trees:
$$s=a_1s_1+\cdots+a_qs_q,\;\;\;t=b_1t_1+\cdots+b_rt_r.$$
Since $(s)^{\alpha_{\Gamma}}\sim_k (t)^{\alpha_{\Gamma}},$ the two sets of pairs
$$\{(a_j,[(s_j)^{\alpha_{\Gamma}}]_{\sim_k}):1\leq j\leq q\}, \{(b_j,[(t_j)^{\alpha_{\Gamma}}]_{\sim_k}):1\leq j\leq r\}$$
are equal.  As we argued in the previous lemma, if
$$(s_j)^{\alpha_{\Gamma}}\sim_k (t_{j'})^{\alpha_{\Gamma}},$$
then nonconfusion for $\beta$ makes their values under $\beta$ equal.  Let us denote this common value by $h.$  We then have, for $a\in A,$
$$\alpha(as_j)=\alpha(a)(h,h')=(\beta(a)h,\gamma(a,h)h').$$
We similarly have
$$\alpha(at_{j'})=\alpha(a)(h,h'')=(\beta(a)h,\gamma(a,h)h''),$$
and the two values are the same by 1-definiteness of $\gamma.$  Thus
$$\{\alpha(a_js_j):1\leq j\leq q\}=\{\alpha(b_jt_j):1\leq j\leq r\}.$$
 We  get $\alpha(s_1)=\alpha(s_2)$ by summing over these two sets and using idempotence and commutativity of addition.
\end{proof}

\subsection{Decidability}\label{section.decidability}
In this section we show that we can effectively determine if a given forest algebra homomorphism $\alpha:A^{\Delta}\to (H,V)$ is nonconfusing. Our method is a variant of the one given by Bojanczyk and Walukiewicz for the binary tree case.~\cite{BW1}

We suppose that $\alpha$ is onto, and that $\Gamma\subseteq H$ is a reachability class.  We first describe an algorithm that constructs a finite sequence of subsets 
$B_0,B_1,B_2,\ldots$ of $\Gamma\times\Gamma.$ 

\begin{itemize}
\item Set 
$$B_0=\{(h,h')\in\Gamma\times\Gamma: h\neq h'\}.$$
\item For $j=0,1,\ldots,$
\begin{itemize}
\item Initially, set 
$B_{j+1}=\{(\alpha(a)h,\alpha(a)h'):a\in A,  (\alpha(a)h,\alpha(a)h')\in B_0,(h,h')\in B_j\}.$
\item If there exist $(h,h')\in B_{j+1}$ and $g\in H$ such that $(h+g,h'+g)\in B_0,$ add $(h+g,h'+g)$ to $B_{j+1}.$
\item If there exist $(h,h'),(g,g')\in B_{j+1}$ such that $(h+g,h'+g')\in B_0,$ add $(h+g,h'+g')$ to $B_{j+1}.$
\item Repeat the preceding two steps until no new elements can be added to $B_{j+1}.$
\end{itemize}
\end{itemize}

Since there are at most $2^{|H|^2}$ different possibilities for the $B_i,$ the algorithm will eventually cycle, so we terminate the execution as soon as we find some $B_i=B_j$ for $i<j.$  In fact, we will see below that the algorithm requires considerably less time and storage than this crude analysis suggests. We prove the following crucial property of the algorithm:

\medskip\begin{thm}\label{thm.algorithmcorrect} Let $h,h'\in \Gamma$ with $h\neq h',$ and let $k\geq 0.$ Then $(h,h')\in B_k$ if and only if there exist $s,t\in H_A$ with
$(s)^{\alpha_{\Gamma}}\sim_k (t)^{\alpha_{\Gamma}},$ $\alpha(s)=h,$ and $\alpha(t)=h'.$ We also have $B_{k+1}\subseteq B_k$ for all $k\geq 0.$
\end{thm}
\begin{proof}Let $C_k$ denote the set of pairs $(h,h')$ satisfying the condition in the statement of the theorem.  We will prove by induction on $k$ that $B_k=C_k$ for all $k\geq 0.$ The case $k=0$ is trivial, so assume that $B_k=C_k$ for some $k\geq 0.$  We show $B_{k+1}=C_{k+1}.$

We first prove $B_{k+1}\subseteq C_{k+1}$ by induction on the construction of $B_{k+1}.$  A pair added to $B_{k+1}$ at the initial step of the construction has the form $(\alpha(a)h,\alpha(a)h')\in B_0.$ where $a\in A,$ $(h,h')\in B_k.$  By the induction hypothesis, $(h,h')\in C_k,$ so there is a pair of forests $s,t$ with $(s)^{\alpha_{\Gamma}}\sim_k (t)^{\alpha_{\Gamma}}$ and $\alpha(s)=h,\alpha(t)=h'.$ 
Since $\alpha(s),\alpha(t)\in\Gamma,$ we have
$$(as)^{\alpha_{\Gamma}}=(a,\infty)\cdot (s)^{\alpha_{\Gamma}}, (at)^{\alpha_{\Gamma}}=(a,\infty)\cdot (t)^{\alpha_{\Gamma}},$$
so that $(as)^{\alpha_{\Gamma}}\sim_{k+1}(at)^{\alpha_{\Gamma}}.$  Thus $(\alpha(a)h,\alpha(a)h')\in C_{k+1}.$

Now suppose that a pair $(h,h')$ is added to $B_{k+1}$ after this initial step.  Then we have either 
$h=h_1+g,$ $h'=h_2+g,$ where $(h_1,h_2)\in B_{k+1}$ was added at an earlier step; or $h=h_1+g_1,$
$h'=h_2+g_2,$ where $(h_1,h_2),(g_2,g_2)\in B_{k+1}$ were added at earlier steps.  In the first case, the hypothesis of induction by construction gives $(h_1,h_2)\in C_{k+1},$ so there exists a pair $s,t$ of forests with $(s)^{\alpha_{\Gamma}}\sim_{k+1} (t)^{\alpha_{\Gamma}},$ and $\alpha(s)=h_1,\alpha(t)=h_2.$  Let $u$ be any forest  such that $\alpha(u)=g.$  Then 
$$(s+u)^{\alpha_{\Gamma}}=(s)^{\alpha_{\Gamma}}+(u)^{\alpha_{\Gamma}}\sim_{k+1} (t)^{\alpha_{\Gamma}}+(u)^{\alpha_{\Gamma}}=(t+u)^{\alpha_{\Gamma}},$$
and $\alpha(s+u)=h,$ $\alpha(t+u)=h',$ so $(h,h')\in C_{k+1}.$ We argue the second case similarly, replacing $u$ by a second pair of forests $s',t'$  with $(s')^{\alpha_{\Gamma}}\sim_{k+1} (t')^{\alpha_{\Gamma}}$ mapping to $g_1,g_2.$ This shows $B_{k+1}\subseteq C_{k+1}.$

We now prove the opposite inclusion $C_{k+1}\subseteq B_{k+1}$.  It is sufficient to show that whenever $(s)^{\alpha_{\Gamma}}\sim_{k+1} (t)^{\alpha_{\Gamma}},$ and $\alpha(s),\alpha(t)\in{\Gamma},$  then either $\alpha(s)=\alpha(t),$ or $(\alpha(s),\alpha(t))\in B_{k+1}.$ 
We write
$$s=a_1s_1+\cdots +a_mt_m, t=a'_1t_1+\cdots +a'_nt_n.$$
For each $i=1,\ldots,m$ there is $j=1,\ldots,n$ such that $(a_is_i)^{\alpha_{\Gamma}}\sim_{k+1} (a'_jt_j)^{\alpha_{\Gamma}}$ (which implies in particular that $a_i=a'_j$) and vice-versa.  Thus by duplicating and reordering terms, we can assume that
$$s=a_1s_1+\cdots + a_ms_m,$$
$$t=a_1t_1 +\cdots + a_mt_m,$$
where $(a_js_j)^{\alpha^{\Gamma}}\sim_{k+1} (a_jt_j)^{\alpha^{\Gamma}}$ for each $j.$  We show by induction on $j$ that for each pair of partial sums
$$s^{(j)}=a_1s_1+\cdots + a_js_j,$$
$$t^{(j)}=a_1t_1+\cdots +a_jt_j,$$
we either have $\alpha(s^{(j)})=\alpha(t^{(j)}),$ or $(\alpha(s^{(j)}),\alpha(t^{(j)}))\in B_{k+1}.$  This is true for $j=0,$ since we can take $s^{(0)}=t^{(0)}=0,$ the empty forest. 
 Suppose it holds for some $j\geq 0$ and consider the partial sums $s^{(j+1)}, t^{(j+1)}.$  There are several cases to consider, depending on whether or not $\alpha(s^{(j)})=\alpha(t^{(j)}),$ $\alpha(s^{(j+1)})=\alpha(t^{(j+1)}),$ and $\alpha(a_{j+1}s_{j+1})=\alpha(a_{j+1}t_{j+1}).$
We will treat in detail the case where we have inequality for all three of these pairs; the other cases are proved similarly, and are easier.  Let $h_1=\alpha(s^{(j)}),$ $h_2=\alpha(t^{(j)}).$ 
Since we assume $h_1\neq h_2,$   the induction hypothesis on $j$ gives $(h_1,h_2)\in B_{k+1}.$ Since $\alpha(a_{j+1}s_{j+1})\neq\alpha(a_{j+1}t_{j+1}),$ we must have 
$(a_{j+1}s_{j+1})^{\alpha_{\Gamma}}=(a_{j+1},\infty)(s_{j+1})^{\alpha_{\Gamma}},(a_{j+1}t_{j+1})^{\alpha_{\Gamma}}=(a_{j+1},\infty)(t_{j+1})^{\alpha_{\Gamma}}.$
This is because $\sim_{k+1}$-equivalence implies the root nodes must be equal, and if the common value was $(a_{j+1},h)$ for some $h>{\Gamma},$ then we would get $\alpha(a_{j+1}s_{j+1})=\alpha(a_{j+1})h=\alpha(a_{j+1}t_{j+1}),$ contrary to assumption. Thus $\alpha(s_{j+1}),\alpha(t_{j+1})\in\Gamma.$ By the induction hypothesis (on $k$), $(\alpha(s_{j+1}),\alpha(t_{j+1}))\in B_k,$ so we get $(\alpha(a_{j+1}s_{j+1}),\alpha(a_{j+1}t_{j+1})\in B_{k+1}.$ Finally
$$\alpha(s^{(j+1)})=h_1+\alpha(a_{j+1}s_{j+1}),\alpha(t^{(j+1)})=h_2+\alpha(a_{j+1}t_{j+1}),$$
so $(\alpha(s^{(j+1)}),\alpha(t^{(j+1)}))\in B_{k+1}.$  This completes the proof that $B_k=C_k$ for all $k.$

Since $\sim_{k+1}$ refines $\sim_{k},$ we obtain $C_{k+1}\subseteq C_k,$ and thus $B_{k+1}\subseteq B_k$ for all $k\geq 0.$

\end{proof}

Because $B_{k+1}\subseteq B_k,$  the algorithm will terminate when for some $k,$ either $B_k=\emptyset$ or  $B_{k+1}=B_k\neq\emptyset.$ Theorem~\ref{thm.algorithmcorrect} then tells us that if $B_k=\emptyset$ for all reachability classes $\Gamma,$ then $\alpha$ is nonconfusing.  Otherwise, there is some $\Gamma$ for which $C_k$ is nonempty for every $k,$ and thus $\alpha$ cannot be nonconfusing.
 
 Since  $|B_k|$ strictly decreases at each execution of the outer loop of the algorithm, this loop will not 
be executed more than $|\Gamma|^2<|H|^2$ times.  Computing each set $B_j$ in the inner loop also takes time that is polynomial in $|H|$ and $|A|,$ assuming that we have access to the table of operations in $(H,V)$ and the values of $\alpha(a)$ for $a\in A.$  Finally, given the table of addition in $H$ and the action of letters of $A$ on $H,$ we can construct the graph of the reachability order and compute its strongly connected components in time polynomial in $|H|$ and $|A|,$ so the running time of the entire algorithm is polynomial in $|H|$ and $|A|$. 

We summarize these observations as follows:
\medskip\begin{thm}\label{thm.summary}~
\begin{itemize}
\item[\it (a)] A homomorphism $\alpha:A^{\Delta}\to (H,V)$ is an $(\mathsf{EF},\mathsf{EX})$-homomorphism if and only if it is nonconfusing. Furthermore, $\alpha$ is nonconfusing, then it is nonconfusing with parameter $|H|^2.$ 
\item[\it (b)] We can determine in time polynomial in $(|A|+|H|)$ whether a given $\alpha$ is an $(\mathsf{EF}+\mathsf{EX})$-homomorphism.
\end{itemize}
\end{thm}

\section{Results}\label{sec.con}

Using the wreath product characterizations of $\mathsf{EF}$-algebras, $\mathsf{EX}$-homomorphisms, and $(\mathsf{EF,EX})$-homomorphisms of the previous three sections, we get:
\medskip\begin{thm} Let $A$ be a finite alphabet, and let $L\subseteq H_A.$
\begin{itemize}
\item[\it (a)] $L$ is defined by an $\mathsf{EF}$-formula if and only if $(H_L,V_L)$ is an $\mathsf{EF}$-algebra.
\item[\it (b)] $L$ is defined by an $\mathsf{EX}$-formula if and only if $\mu_L$ is an $\mathsf{EX}$-homomorphism.
\item[\it (c)] $L$  is defined by an $\mathsf{EF+EX}$-formula if and only if $\mu_L$ is an $(\mathsf{EF,EX})$-homomorphism.
\item[\it (d)] There are effective procedures for determining, given a finite tree automaton recognizing $L$, whether $L$ is definable by an  $\mathsf{EF}$-, $\mathsf{EX}$-, or $\mathsf{EF+EX}$-formula, and for producing a defining formula in case one exists.
\end{itemize}
\end{thm}
\begin{proof} The first three assertions are proved similarly; we give the proof for the third one, as it is the most general.  First, suppose $L$ is defined by an $\mathsf{EF+EX}$-formula.  We prove by induction on the depth of nesting of the operators in the formula that $\mu_L$ is an $(\mathsf{EF,EX})$-homomorphism.  The base case is when the depth of nesting is 0.  The only forest formulas with nesting depth 0 are ${\bf T}$ and ${\bf F},$ in which case $L$ is either $H_A$ or $\emptyset,$ and the syntactic forest algebra is trivial.  We now suppose that $L$ is defined by a forest formula $\phi,$ with nesting depth $k>0,$ and that the syntactic morphism of every language defined by a formula of smaller depth is an $(\mathsf{EF,EX})$-homomorphism. We can write $\phi$ as a boolean combination of formulas of the form $\mathsf{EF}\tau$ and $\mathsf{EX}\tau,$ where $\tau$ is a tree formula.  It suffices to show that the syntactic morphisms of the languages defined by $\mathsf{EF}\tau$ and $\mathsf{EX}\tau$ are $(\mathsf{EF,EX})$-homomorphisms: this is because the syntactic morphism of the union or intersection of two languages factors through the direct product of the syntactic morphisms of the two languages, which in turn factors through the wreath product.  By the inductive hypothesis, the languages defined by the forest formulas of $\tau$ are recognized by $(\mathsf{EF,EX})$-homomorphisms, so by Propositions~\ref{prop.efoperator} and \ref{prop.exoperator}, so are the languages defined by $\mathsf{EF}\tau$ and $\mathsf{EX}\tau.$  Thus the syntactic morphisms of these languages are $(\mathsf{EF,EX})$-homomorphisms.

Conversely, suppose the syntactic morphism of $L\subseteq H_A$ is an $(\mathsf{EF,EX})$-homomorphism.  Then $L$ is recognized by a wreath product $\alpha_1\otimes\cdots\otimes\alpha_r,$ where each component homomorphism is either 1-definite or maps into ${\cal U}_1.$  It follows from Propositions~\ref{prop.efoperator} and \ref{prop.exoperator} and induction on $r$ that $L$ is defined by an $\mathsf{EF+EX}$-formula.

We turn to the results about effectively determining definability and producing formulas.  We can construct both the syntactic forest algebra and syntactic morphisms for a language $L$ from any automaton recognizing $L.$  For the case of definability by $\mathsf{EF}$ formulas we only need to verify the identities for $\mathsf{EF}$-algebras.  For $\mathsf{EX}$ formulas, we need to test whether $\mu_L(V^{\mathsf{gu}})$ is a reverse-definite semigroup. A semigroup $S$ is reverse-definite if and only if $es=s$ for all $e,s\in S$ with $e$ idempotent, so this too can be determined effectively. Using the characterization of Theorem~\ref{thm.main} our results in Section~\ref{section.decidability} show that we can effectively determine definability by $\mathsf{EF+EX}$-formulas.  While we have not provided a streamlined algorithm for producing the defining formulas themselves, our proofs of wreath product decompositions are entirely constructive, and the formulas themselves can be derived from the construction of these decompositions along with Propositions~\ref{prop.efoperator} and \ref{prop.exoperator}.
\end{proof}

\bibliographystyle{alpha} 
\bibliography{forestalgebras}

\end{document}